\documentclass[12pt, centerh1]{article}
\RequirePackage[colorlinks,citecolor=blue,urlcolor=blue]{hyperref}
\usepackage{amsmath, amssymb,natbib}
\usepackage{graphicx,bm}
\usepackage{color}
\usepackage{subcaption}
 \usepackage[table]{xcolor}
\usepackage{longtable}
\usepackage{amsthm}
\usepackage[mathscr]{euscript}
\usepackage{relsize}
\newcolumntype{P}[1]{>{\centering\arraybackslash}p{#1}}
\usepackage{rotating}
\usepackage{eurosym}
\usepackage{colonequals}
\usepackage{bbm}

\newcommand{\MM}{\bm{M}}
\newcommand{\MV}{\bm{V}}
\newcommand{\MU}{\bm{U}}

\newtheorem{theorem}{Theorem}

\title{Assessing and Visualizing Matrix Variate Normality} 
\author{\qquad Nikola Po\v cu\v ca \qquad\qquad\ Michael P.B.\ Gallaugher\\ Katharine M.\ Clark  \qquad\qquad Paul D.\ McNicholas}
\date{{\small Department of Mathematics and Statistics, McMaster University, Ontario, Canada.}}

\begin{document}

\maketitle

\begin{abstract}
A framework for assessing the matrix variate normality of three-way data is developed. The framework comprises a visual method and a goodness of fit test based on the Mahalanobis squared distance (MSD). The MSD of multivariate and matrix variate normal estimators, respectively, are used as an assessment tool for matrix variate normality. Specifically, these are used in the form of a distance-distance (DD) plot as a graphical method for visualizing matrix variate normality. In addition, we employ the popular Kolmogorov-Smirnov goodness of fit test in the context of assessing matrix variate normality for three-way data. Finally, an appropriate simulation study spanning a large range of dimensions and data sizes shows that for various settings, the test proves itself highly robust.\\[-10pt]

\noindent\textbf{Keywords}: DD plot, Mahalanobis distance, matrix normality test, matrix variate, three-way data.
\end{abstract}

\section{Introduction}
In recent years, dimensionality and quantity of data have become increasingly large. Three-way data have also become increasingly common, e.g., it is no longer uncommon to measure several quantities for each subject at each time point in a longitudinal study. Many approaches to analyzing two-way data are based on the multivariate normal distribution and much work has been done on assessing the normality of the two-way data. \cite{royston83} extend the univariate Shapiro-Wilk test for normality \citep{shap} to large samples of higher dimension, and \cite{mud} define a null distribution for the Mahalanobis squared distance (MSD) of multivariate normal data. In this paper, we expand these concepts to three-way data by developing approaches for testing matrix variate normality.  Through the use of MSD, our approach builds on the history of existing tests for multivariate normality and extends them into the space of matrix variate normality. In addition, our approach further visualizes matrix variate normality in the context of the relationship between the multivariate and matrix variate normal distributions. 

\section{Background}
\subsection{Matrix Variate Normal Distribution}
Two-way data can be regarded as the observation of $N$ vectors, whereas three-way data can be considered the observation of $N$ matrices. Common examples of three-way data include greyscale images and multivariate longitudinal data. Multivariate distributions have been successfully used in the analysis of two-way data and matrix variate distributions are gaining popularity for the analysis of three way data \citep[e.g.,][]{anderlucci15,MIXskew}. Similar to the multivariate case, the most mathematically tractable matrix variate distribution is the matrix variate normal distribution.
An $r\times c$ random matrix $\mathscr{X}$ comes from a matrix variate normal distribution if its density is of the form
 \newline
\begin{equation}
 \phi_{r\times c}({\bm X}~|~\MM, \MV, \MU ) = \frac{1}{(2\pi)^{ \frac{rc}{2} } |\MV|^{\frac{r}{2}} |\MU|^{\frac{c}{2}}}\exp \left\{- \frac{1}{2}\text{tr}\big(\MV^{-1}(\bm{X}-\MM)^{\top}\MU^{-1}(\bm{X}-\MM) \big) \right\}, 
 \label{eq:pdf}
 \end{equation}
where $\MM$ is the $r \times c$ mean matrix, $\MU$ is the $r \times r$ row covariance matrix, and $\MV$ is the $c \times c$ column covariance matrix.  
Note that the matrix variate normal distribution is related to the multivariate normal distribution via
\begin{equation}
\mathscr{X}\sim \mathcal{N}_{r\times c}(\MM, \MV, \MU) \iff \text{vec}(\mathscr{X})\sim \mathcal{N}_{rc}(\text{vec}(\MM), \MV \otimes \MU)
\label{eq:norm}
\end{equation}
\citep{gupta}, where $\otimes$ denotes the Kronecker product and $\text{vec}(\cdot)$ is the vectorization operator.
Note that there is an identifiability issue with regard to the parameters $\MU$ and $\MV$, i.e., if $c$ is a strictly positive constant, then $$\frac{1}{c}\MV\otimes c\MU=\MV\otimes\MU$$ and so replacing $\MU$ and $\MV$ by $(1/c)\MU$ and $c\MV$, respectively, leaves \eqref{eq:pdf} unchanged. Various different solutions have been proposed to resolve this issue, including setting $\text{tr}(\MU)=r$ or $\MU_{11}=1$ \citep{anderlucci15,MIXskew}.

\subsection{Mahalanobis Squared Distance}
For model interpretability, it is natural to impose a measure of distance between  an observation and a distribution of interest. The Mahalonobis distance is a well-established quantity in the literature \citep{mah}, and \cite{omah} illustrate its application in multivariate outlier detection and goodness of fit. Consider $N$ $p$-dimensional vectors $\bm{y}_1,\ldots,\bm{y}_N$ such that each $\bm{y}_i$ is a realization of a multivariate random variable $\bm{Y} \sim \mathcal{N}_{p}(\bm{\mu},\bm{\Sigma})$. The Mahalanobis squared distance (MSD) for $\bm{y}_i$ is 
\begin{align}
\mathcal{D}(\bm{y}_i, \bm{\mu},\bm{\Sigma}) =  \left(\bm{y}_i - \bm{\mu}   \right)^{\top} \bm{\Sigma}^{-1} \left(\bm{y}_i - \bm{\mu}   \right). \label{md}
\end{align}
It is well known \citep[see][]{mardia} that 
\begin{align}
 \mathcal{D}(\bm{y}_i, \bm{\mu},\bm{\Sigma})\sim \chi^2_p, \label{msd}
\end{align}
where $\chi^2_p$ is chi-square distributed with $p$ degrees of freedom. Now, consider the estimates 
$$\hat{\bm{\mu}}  = \frac{1}{N}\sum_{i=1}^N \bm{y}_i \qquad\text{and}\qquad  \hat{\bm{\Sigma}}  =  \frac{1}{N-1}\sum_{i=1}^N (\bm{y}_i- \hat{\bm{\mu}} )^{\top}( \bm{y}_i- \hat{\bm{\mu}} ),$$
then
\begin{equation}
\frac{N}{(N-1)^2}\mathcal{D}(\bm{y}_i, \bm{\hat{\mu}},\bm{\hat{\Sigma}})  \sim \text{Beta}\left( \frac{p}{2},\frac{N-p-1}{2} \right)\label{msdb}
\end{equation}
\citep{betaDist}. If one considers the estimated distribution for all MSDs within a given sample, then a goodness of fit test naturally presents itself, along with outlier detection and other statistical techniques, in the multivariate setting. 

\subsection{Kolmogorov-Smirnov Goodness of Fit Test}
The Kolmogorov-Smirnov (KS) test determines whether a sample comes from a specified distribution \citep{kstest}. Motivating the test in terms of MSDs, we define the test where the distances compose a sample from the reference distribution.  The empirical distribution function for independent and identically distributed ordered $\mathcal{D}_1,\ldots,\mathcal{D}_N$ is given by
$$ F_N(D) = \frac{1}{N}\sum_{i=1}^N \mathbbm{1}_{(-\infty, D]} \left(\mathcal{D}_i \right),$$
where $$\mathbbm{1}_{(-\infty, D]}=\begin{cases}1 & \text{if } \mathcal{D}_i \leq D,\\ 0 & \text{otherwise}.\end{cases}$$ By definition, the KS statistic for a given cumulative distribution function $F(D)$ is 
$$\mathbb{D}_N = \underset{D}\sup|F_N(D) - F(D) |. $$
This test statistic can be modified to account for two samples, which is appropriate because we are comparing distances from two different estimates which theoretically converge to the same distribution. The two-sample KS test statistic is
$$\mathbb{D}_{N_a,N_b} = \underset{D}\sup|F_{1,N_a}(D) - F_{2,N_b}(D) |,$$
where $N_a$ and $N_b$ are the numbers of observations in samples $a$ and $b$, respectively, and $F_{1,N_a}$ and $F_{2,N_b}$ denote the cumulative distribution functions defined above for the first and second samples, respectively. The two-sample KS test considers the hypotheses: 
\begin{align*}
H_0: &\text{ Both samples come from the same distribution} \\
H_a: & \text{ Both samples do not come from the same distribution.}
\end{align*}
The null hypothesis is rejected at significance level $\alpha$ if 
$$\mathbb{D}_{N_a,N_b} > \sqrt{-\frac{1}{2}\log(\alpha) \left(  \frac{N_a + N_b}{N_a N_b}  \right)  } .$$
This two-sample test will be used to test matrix variate normality by comparing the distributions of the two distances described in Section~\ref{sec:dd}. The main benefit of this approach is that it does not specify the common distribution between the two samples. As a result, it is not as a powerful as the original KS test because it is sensitive against all possible types of differences between two distribution functions \citep{notp}. 

\section{Methodology}
\subsection{Distance-Distance Plot}\label{sec:dd}
We propose a new \textit{post hoc} method for visually assessing the matrix variate structure of a dataset. Consider $N$ $r\times c$ matrices $\bm{X}_1,\ldots,\bm{X}_N$ such that each $\bm{X}_i$ is a realization of a matrix variate random variable $\mathscr{X}\sim \mathcal{N}_{r\times c}(\MM,\MV,\MU)$. Recall the relationship between the matrix variate normal and multivariate normal distributions \eqref{eq:norm}, and
let $\bm{\mu} = \text{vec}(\MM)$ and $\bm{\Sigma} =  \MV \otimes \MU$. Consider the estimates
$$\hat{\bm{\mu}}= \sum_{i=1}^N \frac{\text{vec}(\bm{X}_i)}{N}\qquad\text{and}\qquad\hat{\bm{\Sigma}} = \frac{1}{N-1}\sum_{i=1}^N \left\{\text{vec}(\bm{X}_i) - \hat{\bm{\mu}}\right\}\left\{\text{vec}(\bm{X}_i) - \hat{\bm{\mu}}\right\}^{\top}.$$ Suppose now we calculate the MSD for each observation $\bm{X}_i$ in a given sample as follows: 
\begin{align}
\mathcal{D}(\bm{X}_i, \hat{\bm\mu}, \hat{\bm\Sigma} ) = \left\{\text{vec}(\bm{X}_i) -  \hat{\bm\mu}\right\}^{\top} \hat{\bm\Sigma}^{-1} \left\{\text{vec}(\bm{X}_i) - \hat{\bm\mu}\right\}. \label{msdOrg}
\end{align}
We have that 
\begin{align} \frac{N}{(N-1)^2} \mathcal{D}(\bm{X}_i, \hat{\bm\mu}, \hat{\bm\Sigma} ) \sim \text{Beta}\left(\frac{rc}{2},\frac{N-rc-1}{2}\right). \label{multiMSD}
\end{align}
Moreover, the maximum likelihood estimates for $\MM$,  $\MU$ and $\MV$ are 
\begin{align}
\hat{\MM} &= \frac{1}{N}\sum_{i=1}^N{\bm{X}_i} \label{e1},\\
\hat{\MU} &= \frac{1}{cN}\sum_{i=1}^N \big(\bm{X}_i - \hat{\MM}\big) \hat{\MV}^{-1}\big(\bm{X}_i - \hat{\MM} \big)^{\top} \label{e2},\\
\hat{\MV} &=\frac{1}{rN}\sum_{i=1}^N \big(\bm{X}_i - \hat{\MM}\big)^{\top}  \hat{\MU}^{-1} \big(\bm{X}_i - \hat{\MM}\big) \label{e3}.
\end{align}
Now, let 
\begin{equation}\label{eqn:dm}
\mathcal{D}_M(\bm{X}_i,\MM,\MV,\MU)=\text{tr}\left\{\MU^{-1}(\bm{X}_i - \MM)\MV^{-1}(\bm{X}_i - \MM)^{\top} \right\}.
\end{equation}
This quantity in \eqref{eqn:dm} can be viewed as the matrix variate version of the MSD, and we will use this terminology when referring to this quantity. We now have all the necessary notation to present the following theorem.

\begin{theorem}
If a Kronecker product structure exists for $\bm{\Sigma}$, then 
\begin{align}
\mathcal{D}(\bm{X}_i, \bm{\mu}, \bm{\Sigma} )  &= \mathcal{D}_M(\bm{X}_i, \MM, \MU, \MV)  \label{p3}, \\
 \mathcal{D}_M(\bm{X}_i, \hat{\MM} , \hat{\MU} , \hat{\MV}) &\overset{P}{\longrightarrow} \mathcal{D}_M(\bm{X}_i, \MM , \MU , \MV)\label{p1}, \\ 
 \mathcal{D}(\bm{X}_i, \hat{\bm\mu}, \hat{\bm\Sigma} ) &\overset{P}{\longrightarrow} \mathcal{D}(\bm{X}_i, \bm\mu, \bm\Sigma ) \label{p2}, 
 \end{align}
where $ \overset{P} \longrightarrow$ denotes convergence in probability.
\end{theorem} 

\begin{proof}
Result \eqref{p3} is trivial and follows directly from \eqref{eq:norm}---for completeness, details are given in Appendix~\ref{app:trivial}.
Now, we prove the result in \eqref{p1}. Note that  $$\hat{\MM} = \frac{1}{N}\sum_{i=1}^N{\bm{X}_i}$$ is the MLE for the mean matrix. Because the matrix variate normal distribution is part of the exponential family \citep{gupta}, all MLEs exist and are consistent \citep{gupta2}. Therefore, $ \hat{\MM}  \overset{P}{\to} \MM$. As mentioned previously, the estimates of the scale matrices are unique only up to a strictly positive multiplicative constant; however, their Kronecker product, $\MV\otimes \MU \equalscolon \bm\Sigma $ is unique. Therefore
$$\hat{\MV} \otimes \hat{\MU} \overset{P}{\longrightarrow} \MV \otimes \MU = \bm{\Sigma}.$$
From these two results, and the continuous mapping theorem (stated as Theorem~\ref{thm:cmt} in Appendix~\ref{app:cmt}), we have $$\mathcal{D}_M(\bm{X}_i, \hat{\MM} , \hat{\MU} , \hat{\MV}) \overset{P} \longrightarrow \mathcal{D}_M(\bm{X}_i, \MM , \MU , \MV).$$
Proceeding to the proof of \eqref{p2}, note that the multivariate normal distribution is a member of the exponential family \citep{gupta}. Therefore, the unbiased estimates $\hat{\bm{\mu}}$ and $\hat{\bm{\Sigma}}$ converge in probability to the true parameters $\bm\mu$ and $\bm \Sigma$, respectively. From the continuous mapping theorem, it follows that
$$\left\{\text{vec}(\bm{X}_i) -  \hat{\bm\mu}\right\}^{\top} \hat{\bm\Sigma}^{-1} \left\{\text{vec}(\bm{X}_i) -  \hat{\bm\mu}\right\} 
 \overset{P}{\longrightarrow}   \left\{\text{vec}(\bm{X}_i) -  \bm{\mu}\right\}^{\top} \bm{\Sigma}^{-1} \left\{\text{vec}(\bm{X}_i) - \bm{\mu}\right\},$$
i.e., $$\mathcal{D}(\bm{X}_i, \hat{\bm\mu}, \hat{\bm\Sigma})
 \overset{P}{\longrightarrow} \mathcal{D}(\bm{X}_i, \bm\mu, \bm\Sigma).$$
\end{proof}

From of the results in Theorem 1, it seems useful to visualize matrix variate normality by comparing the estimated MSDs. Consider a plot of $\mathcal{D}$ versus $\mathcal{D}_M$, which we will refer to as the distance-distance (DD) plot. The DD plot is a scatter plot of the Mahalanobis distances using the estimated parameters from the multivariate and matrix variate normal distribution, respectively. Figure~\ref{ddplotEx} illustrates the visual approach to determining the matrix variate normal structure. On the left, we have the DD plot for a matrix variate normal structure with a red line at $\mathcal{D}= \mathcal{D}_M$ for reference. Note that the distances lie roughly along the line with little variability between the multivariate and matrix variate MSDs. On the right side, however, we have that the MSDs exhibit more variability and do not lie along the reference line---this because the data were simulated from a strictly multivariate normal distribution, i.e.,  without a Kronecker product covariance structure. 
\begin{figure}[!htb]
\begin{center}
\includegraphics[scale=0.38]{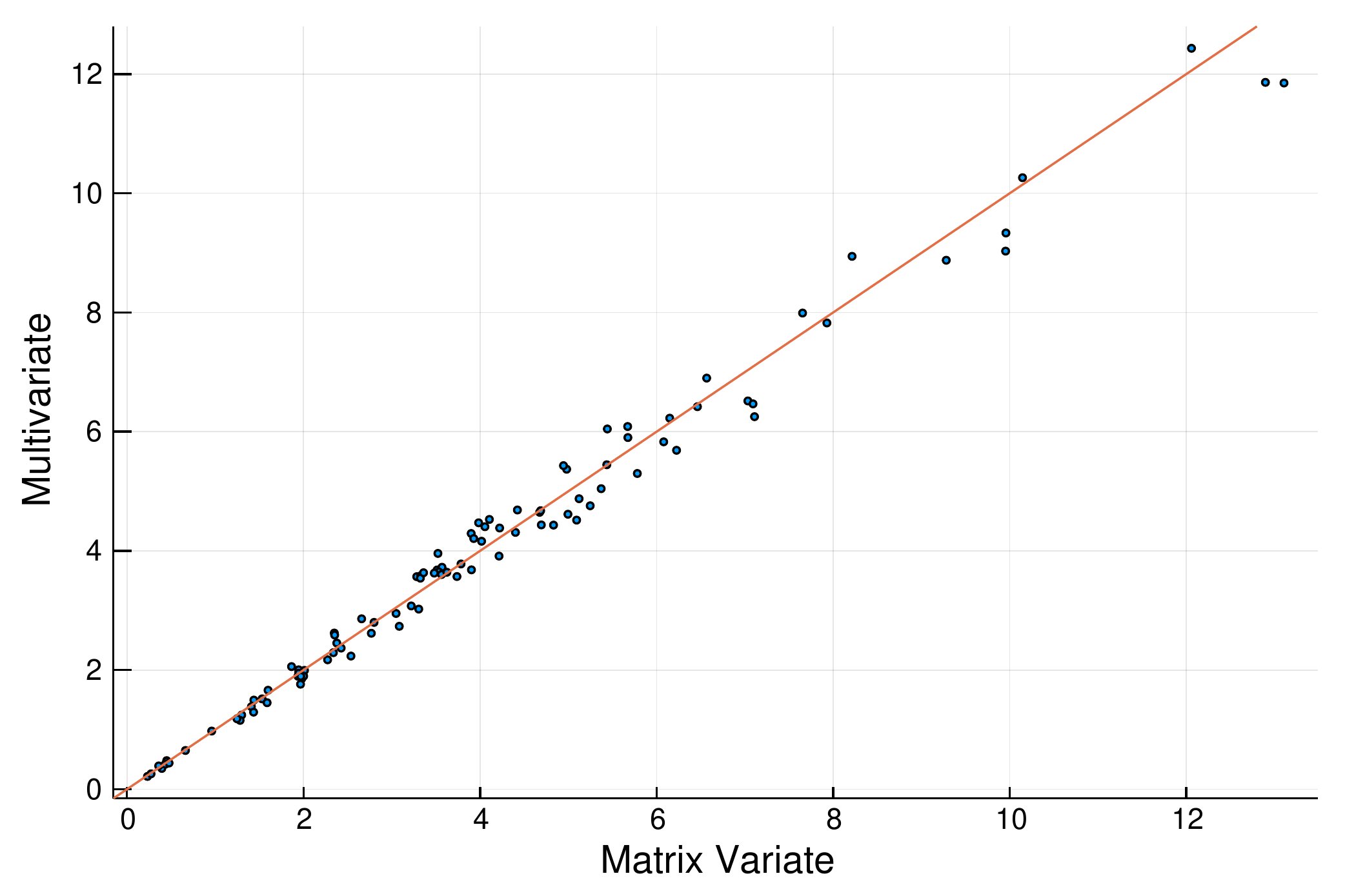}\ \ \ 
\includegraphics[scale=0.38]{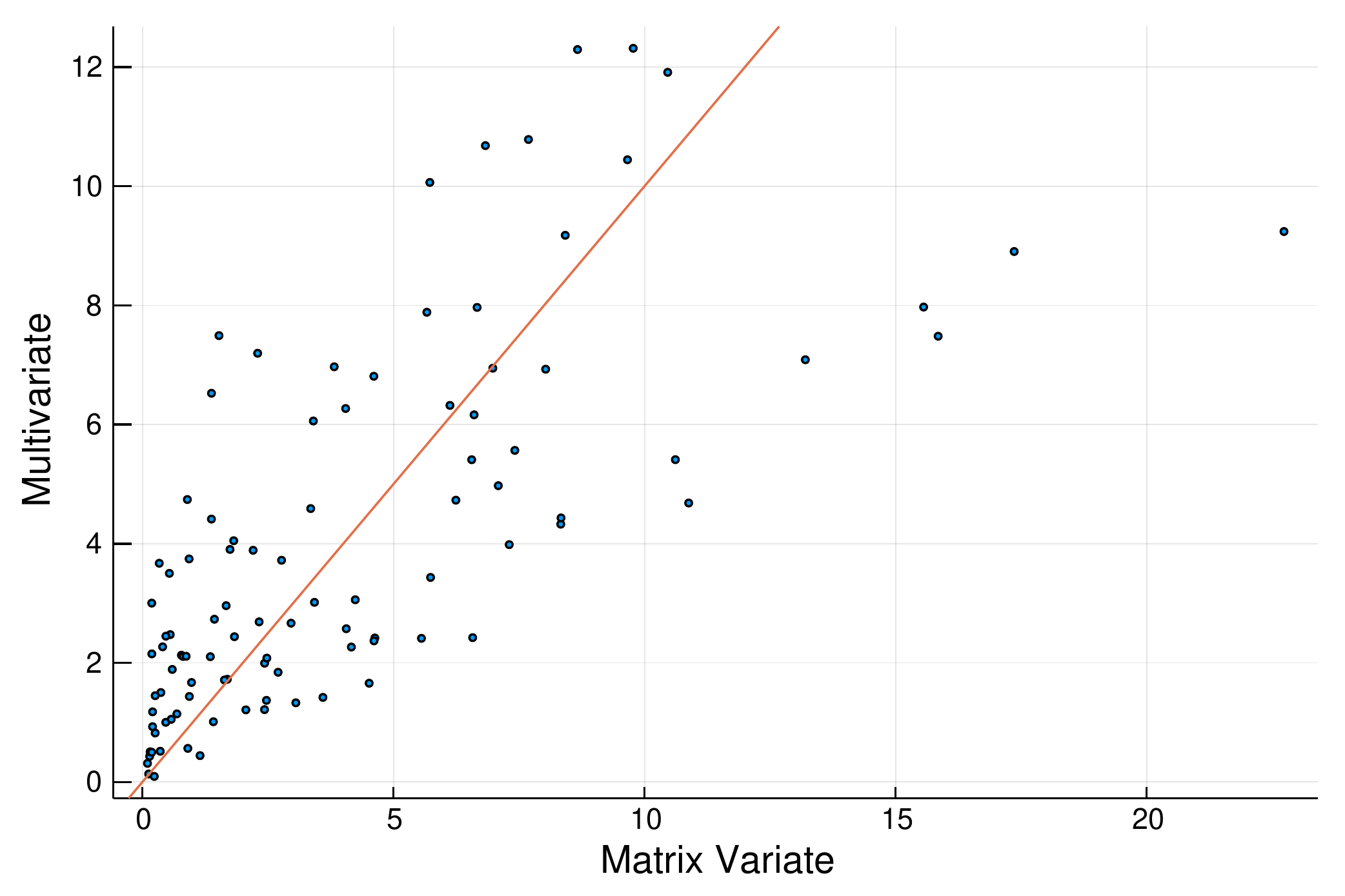}
\caption{DD plots for simulated data ($N=100,p=4$) with randomly chosen mean and variance parameters, indicating the presence (left) and absence (right) of a matrix variate normal structure, i.e., of a Kronecker product covariance structure in the multivariate case.}
\label{ddplotEx}
\end{center}
\end{figure}

\subsection{Test for Matrix Variate Normality}
Once the distances have all been calculated, we can perform the KS test at significance level~$\alpha$. The test for matrix variate normality is defined as 
\begin{align*}
H_0: & \text{ $\mathcal{D}$ and $\mathcal{D}_M$ come from the same distribution.} \\
H_a: & \text{ $\mathcal{D}$ and $\mathcal{D}_M$  do not come from the same distribution.}
\end{align*}
The null hypothesis is rejected if 
$$\mathbb{D}_{N,N}  > \sqrt{-\frac{1}{2}\log(\alpha) \left(  \frac{2N}{N^2}\right)}.$$ 
Because both distances are estimated from the  same sample, they must theoretically converge to the same distribution. Rejection of the null hypothesis indicates that a Kronecker product covariance structure is not present, which violates \eqref{eq:norm}. 

\section{Simulation Study}
\subsection{Overview}
The methodology was implemented in the Julia programming language  \citep{julia,mcnicholas19} and is available within the \texttt{MatrixVariate.jl} package \citep{pocuca19}. 
The simulation study is divided into two separate sections. The first section investigates the effect of dimensionality and sample size on DD plots. Within the first set of investigations, we look at the effect of dimensionality while keeping the number of observations constant. Then, we investigate the effect of sample size while keeping dimensionality constant. The second section investigates the effect on the KS-test performance under the same simulation scheme.  

\subsection{Simulation study for DD plots}
We investigate the effect of dimensionality on our approach for visualizing matrix variate normality. The number of observations is set at $N=1000$ while dimension takes the values $p \in \{4,100,400\}$.  Figure~\ref{fig:dimvary} shows the DD plots for each case, with each row having a different dimension $p$. Figures~\ref{fig:ddplot2M}, \ref{fig:ddplot10M} and~\ref{fig:ddplot20M} show the DD plots for data with a matrix variate normal structure, and Figures~\ref{fig:ddplot2}, \ref{fig:ddplot10} and~\ref{fig:ddplot20} correspond to data that are strictly multivariate normal and not matrix variate normal.  As expected, in Figures~\ref{fig:ddplot2M} and \ref{fig:ddplot10M}, the matrix variate and multivariate MSDs coincide with one another and, as one would expect, the variability about the reference line increases with dimensionality. This is consistent with null distribution from \eqref{msdb} as the spread of the MSDs increases with dimensionality. In Figure~\ref{fig:ddplot20M}, when $p=400$, the plot becomes skewed and no longer exactly follows the reference line. This is not particularly surprising because $p$ is now relatively large relative to $N$. In contrast, Figures~ \ref{fig:ddplot2}, \ref{fig:ddplot10}, and \ref{fig:ddplot20} show highly variable and random MSDs. As the data for the plots on the right are strictly multivariate normal and not matrix variate normal, the MSDs should not and do not coincide with one another.\begin{figure*}[!htb]
\centering
    \begin{subfigure}[t]{0.45\textwidth}
        \centering
         \caption{$N=1000, p=4$}
        \includegraphics[width=\textwidth]{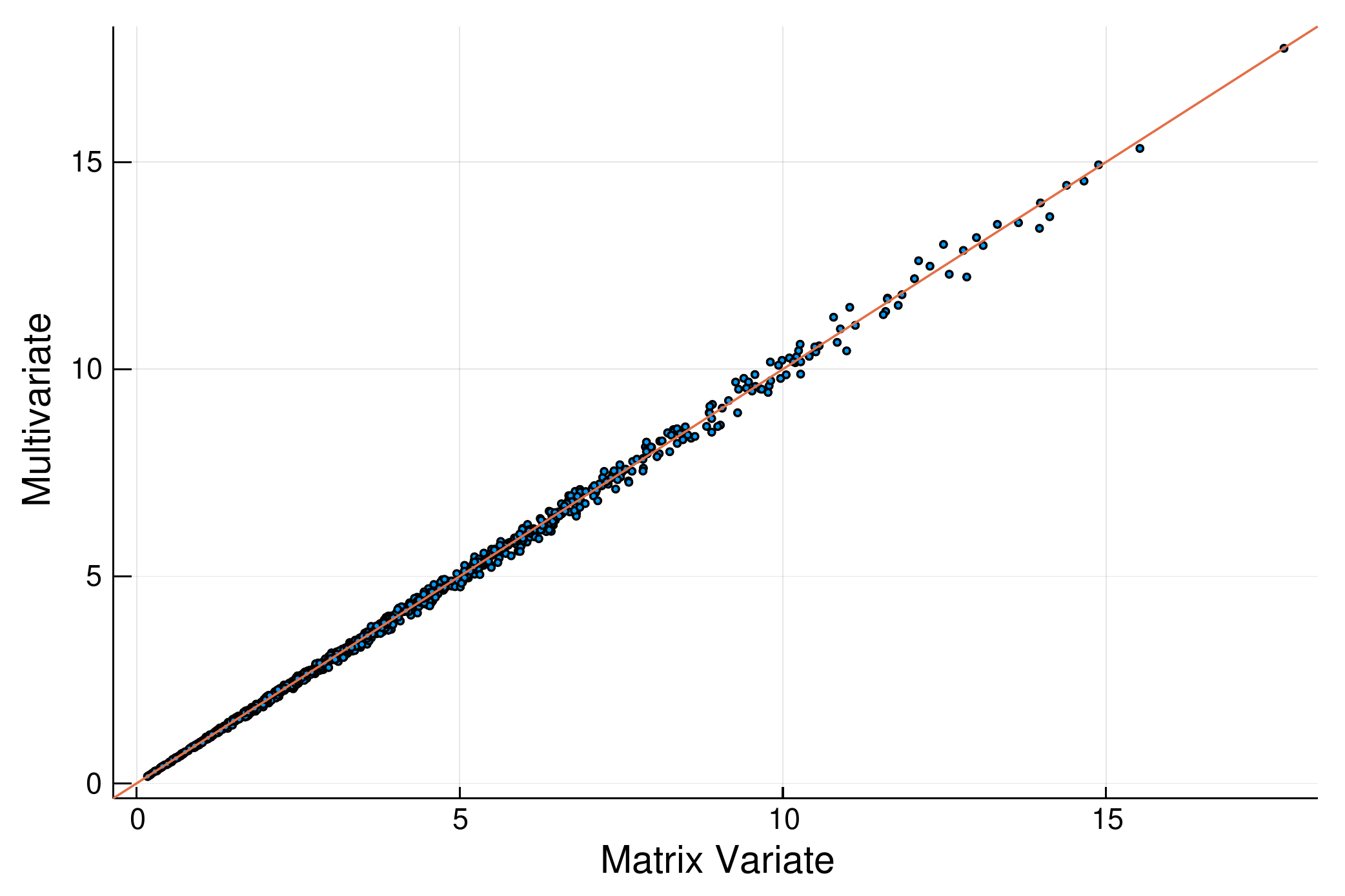}
        \label{fig:ddplot2M}
    \end{subfigure}%
    ~ 
    \begin{subfigure}[t]{0.45\textwidth}
        \centering
        \caption{$N=1000, p=4$}
        \includegraphics[width=\textwidth]{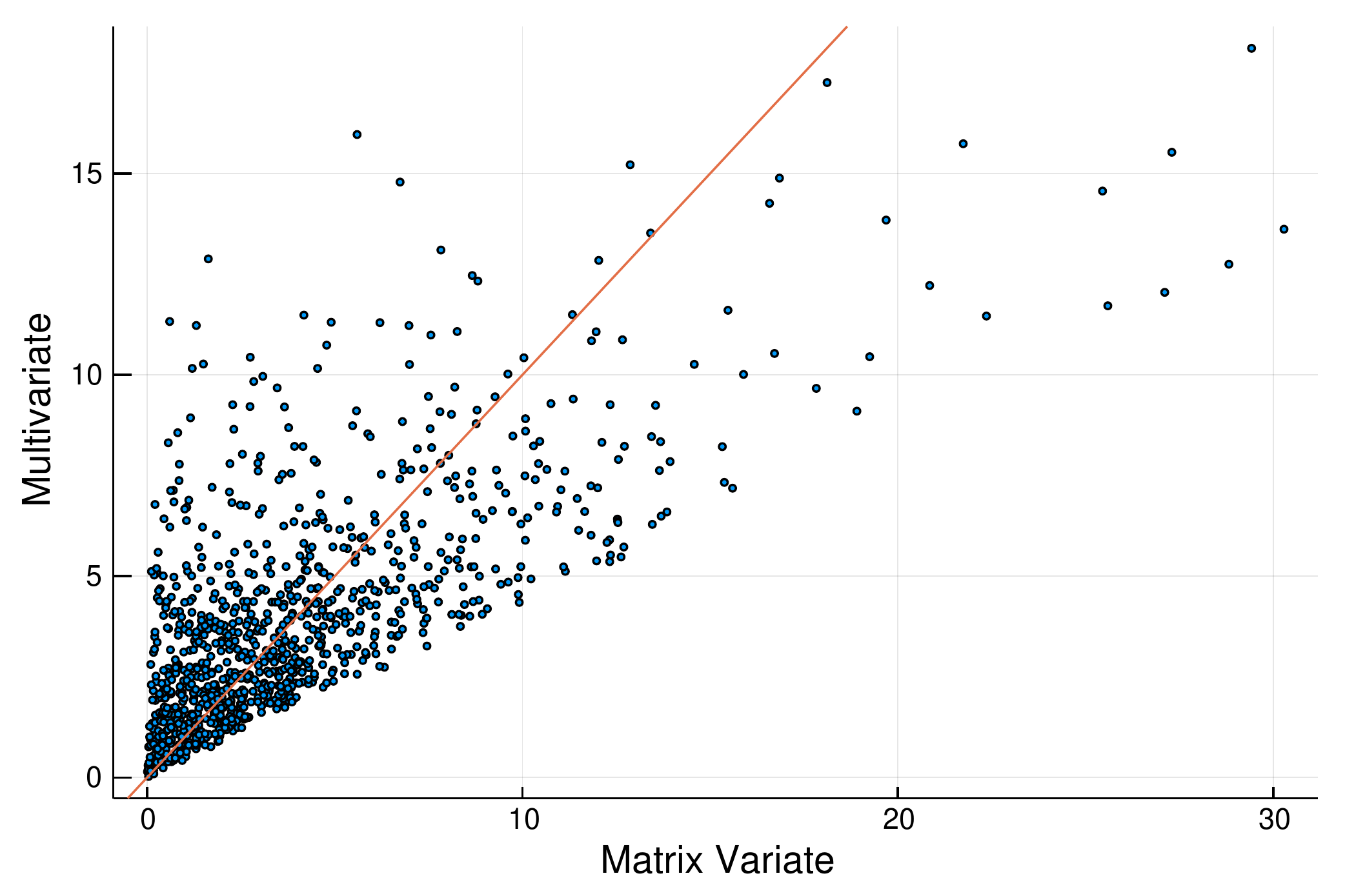}
        \label{fig:ddplot2}
    \end{subfigure}
    \vspace{-0.15in}
    
    \begin{subfigure}[t]{0.45\textwidth}
        \centering
        \caption{$N=1000, p=100$}
        \includegraphics[width=\textwidth]{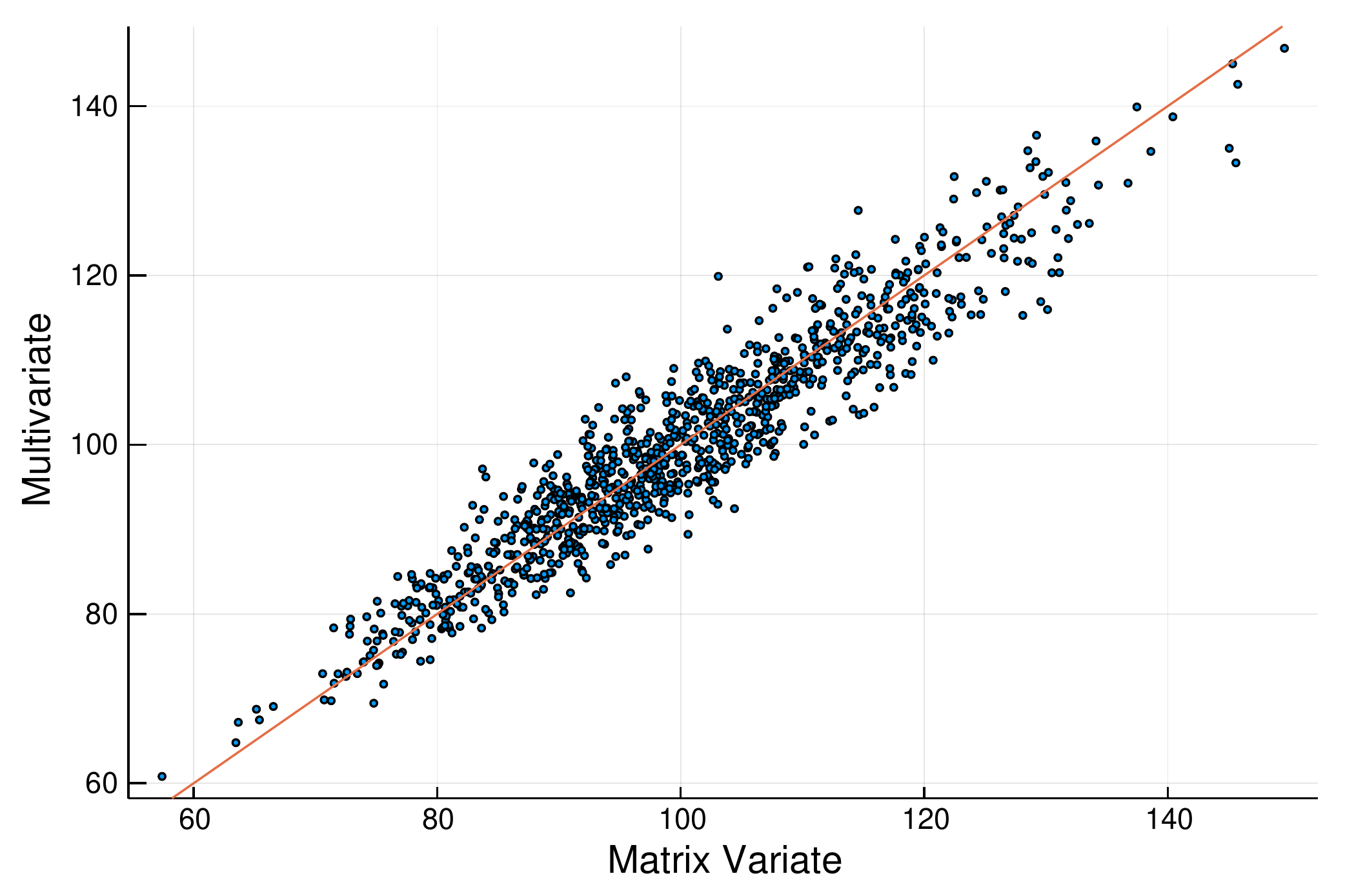}
         \label{fig:ddplot10M}
    \end{subfigure}%
    ~ 
    \begin{subfigure}[t]{0.45\textwidth}
        \centering
        \caption{$N=1000, p=100$}
        \includegraphics[width=\textwidth]{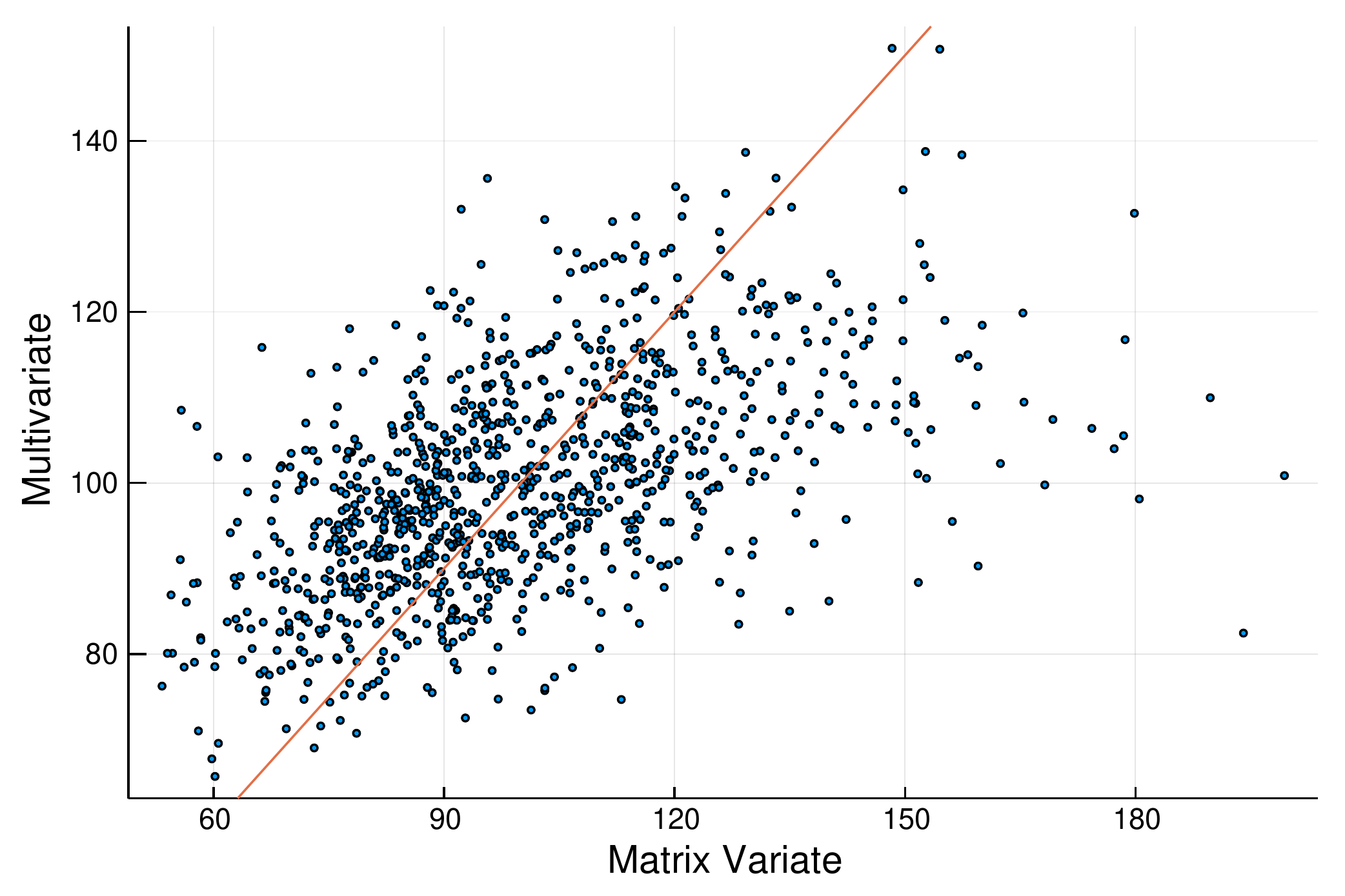}
        \label{fig:ddplot10}
    \end{subfigure}
    \vspace{-0.15in}
    
    \begin{subfigure}[t]{0.45\textwidth}
        \centering
        \caption{$N=1000, p=400$}
        \includegraphics[width=\textwidth]{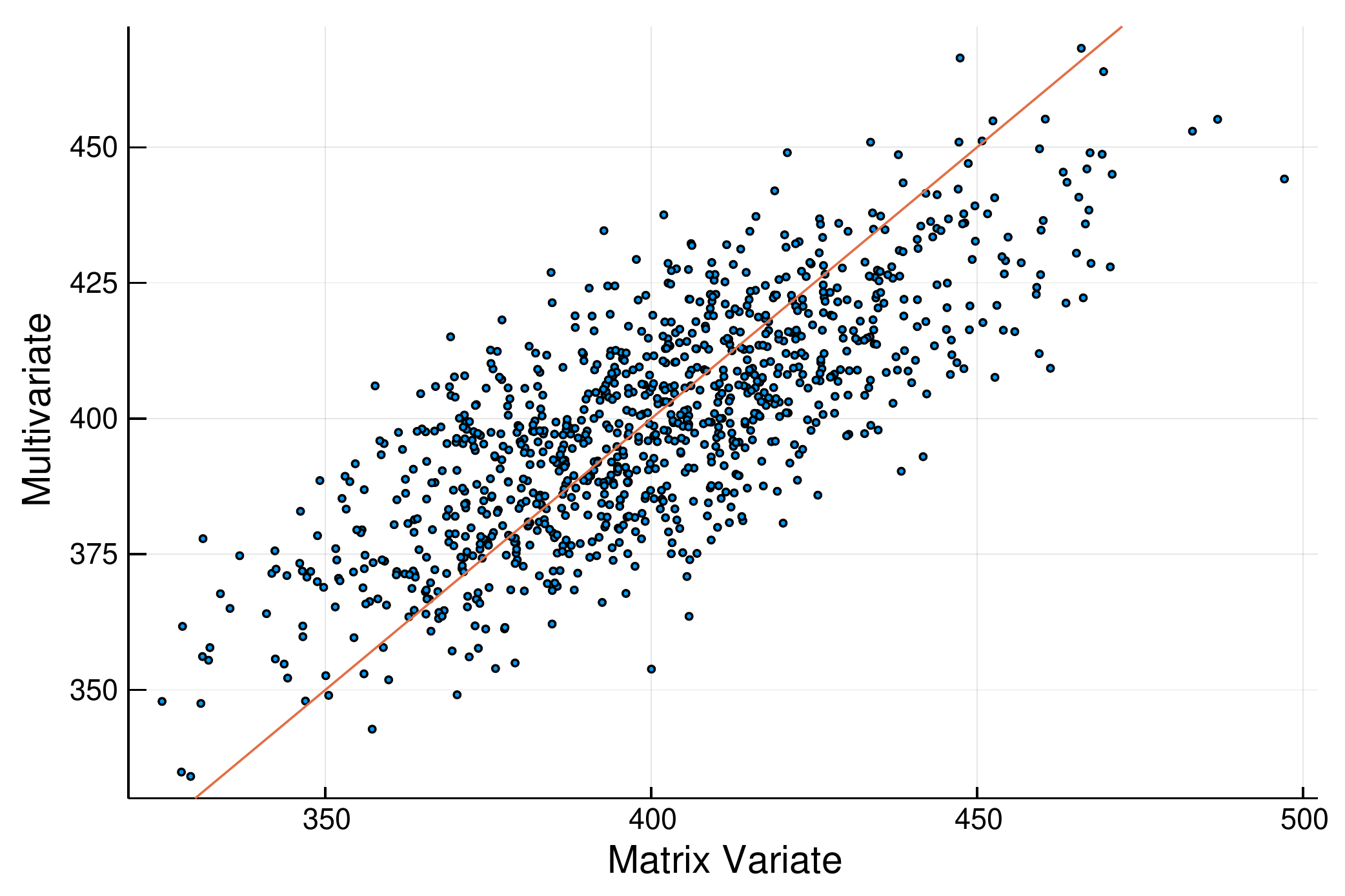}
        \label{fig:ddplot20M}
    \end{subfigure}%
    ~ 
    \begin{subfigure}[t]{0.45\textwidth}
        \centering
        \caption{$N=1000, p=400$}
        \includegraphics[width=\textwidth]{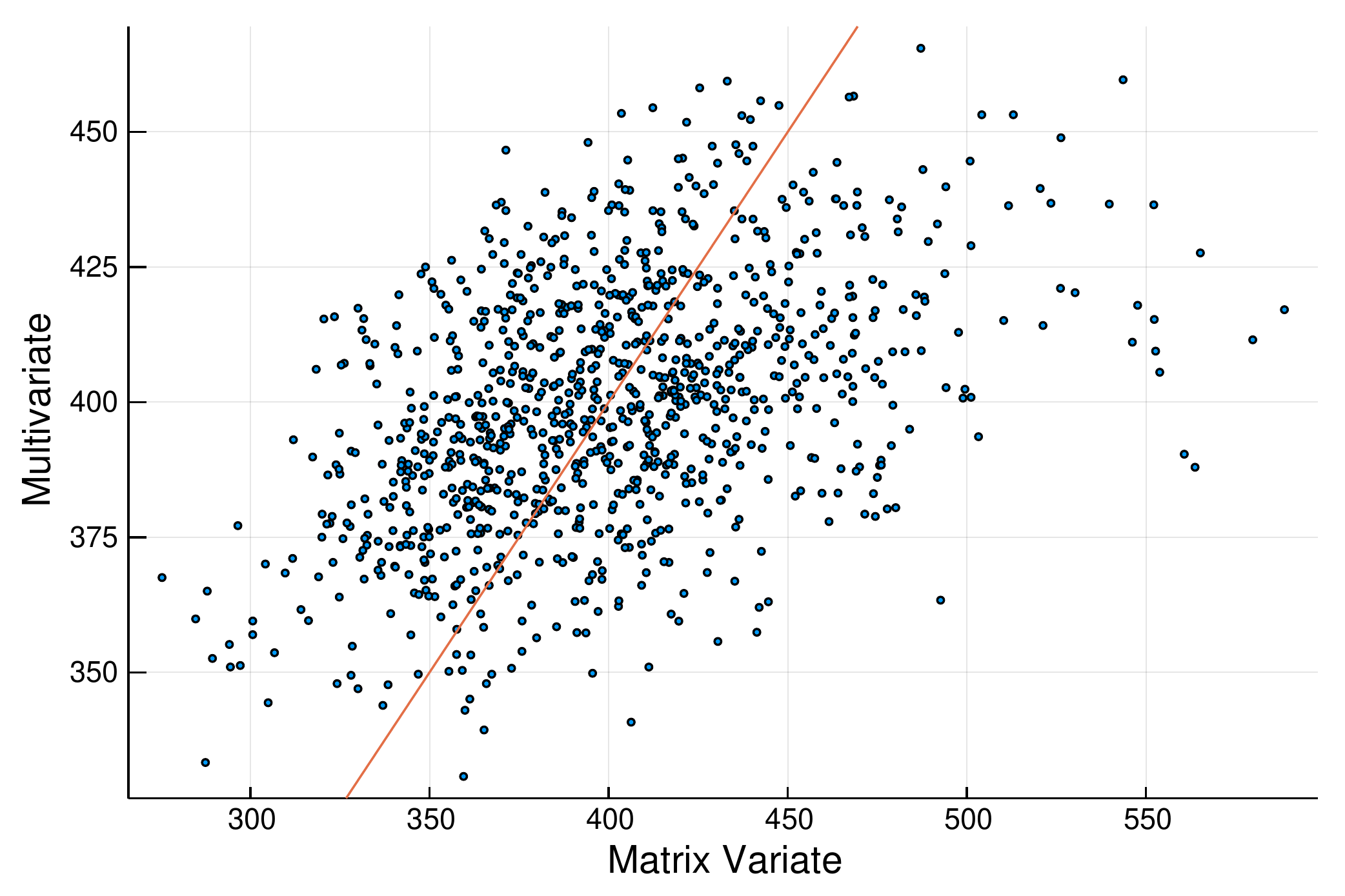}
        \label{fig:ddplot20}
    \end{subfigure}
        \vspace{-0.15in}

    \caption{DD plots for simulated data with $N=1000$ and $p\in \{4,100,400\}$ where matrix variate normal structure is present (left) or absent (right).}
    \label{fig:dimvary}
\end{figure*}

The second part varies the sample size while keeping the dimension constant. In these simulations, we set the dimension of the generated random vectors to $p=100$. Figure~\ref{fig:Nvary} displays an array of DD plots for matrix normal and multivariate normal datasets when $p=100$ and $N\in \{500,2000,10000\}$. Similar to the first investigation, Figures~\ref{fig:ddplot500M}, \ref{fig:ddplot2000M} and~\ref{fig:ddplot10000M} represent datasets which are matrix variate normal, and Figures~\ref{fig:ddplot500}, \ref{fig:ddplot2000} and~\ref{fig:ddplot10000} represent datasets which are multivariate normal but not matrix variate normal. The plots demonstrate that the distances from matrix variate normal data follow the reference line with some random variability, which decreases as dimension increases. When data are strictly multivariate normal and the matrix variate structure is absent, the variability is large, the distances are skewed, and the MSDs diverge from the reference line. This indicates that our approach is highly consistent for a large variety of sample sizes while keeping dimensionality constant. 
\begin{figure*}[!htb]
\centering
    \begin{subfigure}[t]{0.45\textwidth}
        \centering
        \caption{$N=500, p=100$}
        \includegraphics[width=\textwidth]{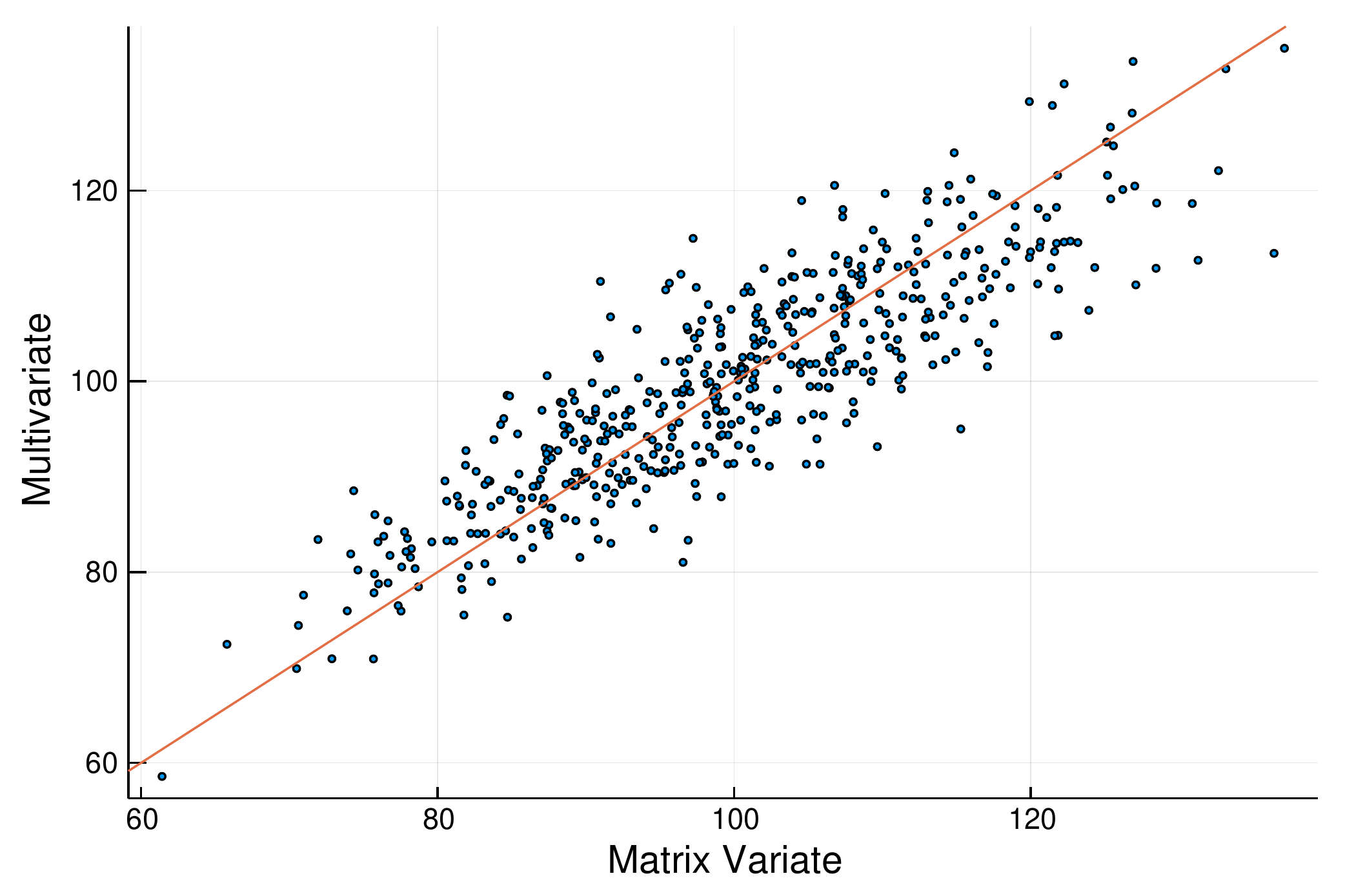}
        \label{fig:ddplot500M}
    \end{subfigure}%
    ~ 
    \begin{subfigure}[t]{0.45\textwidth}
        \centering
        \caption{$N=500, p=100$}
        \includegraphics[width=\textwidth]{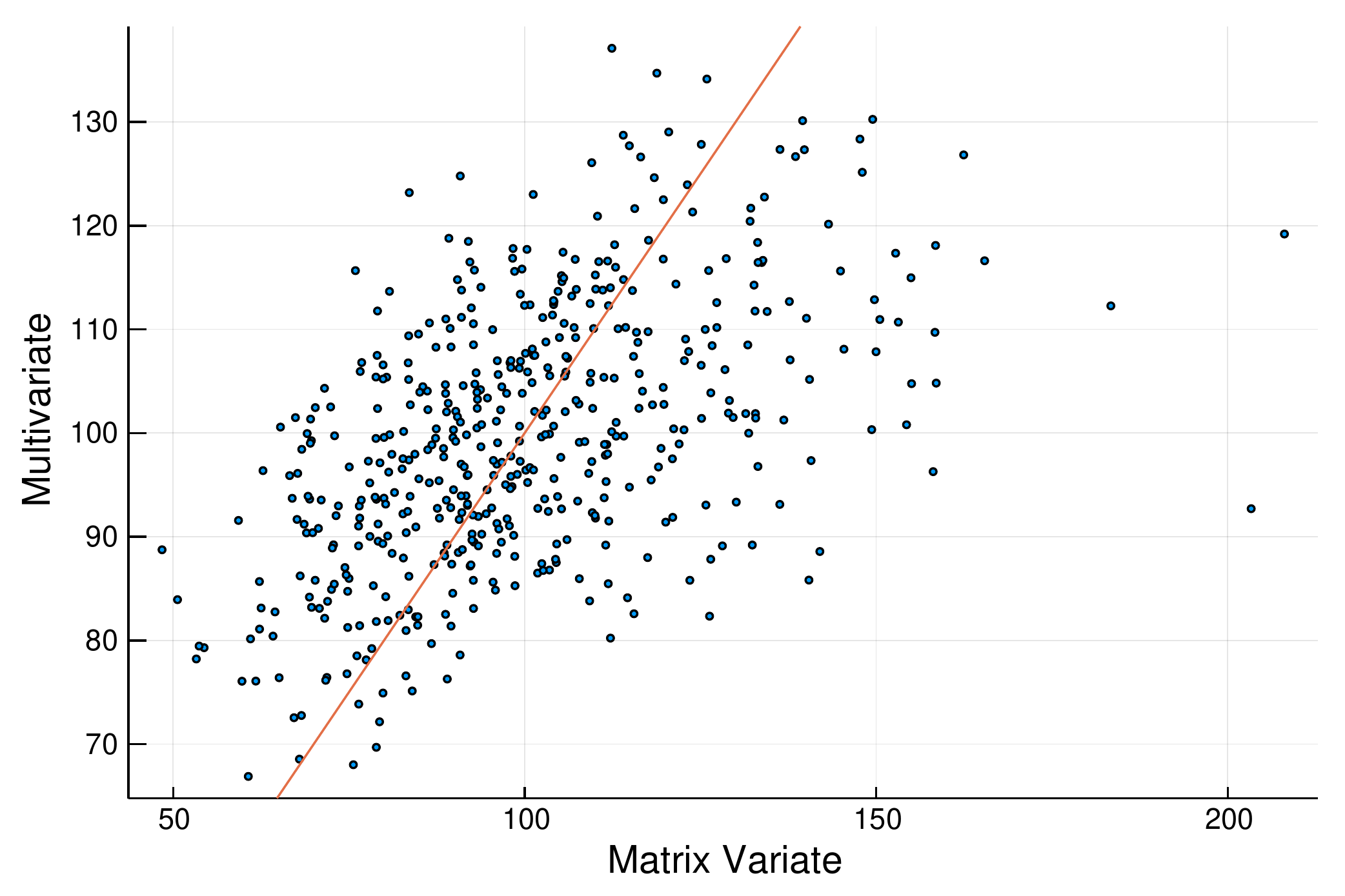}
        \label{fig:ddplot500}
    \end{subfigure}
    \vspace{-0.15in}
    
    \begin{subfigure}[t]{0.45\textwidth}
        \centering
        \caption{$N=2000, p=100$}
        \includegraphics[width=\textwidth]{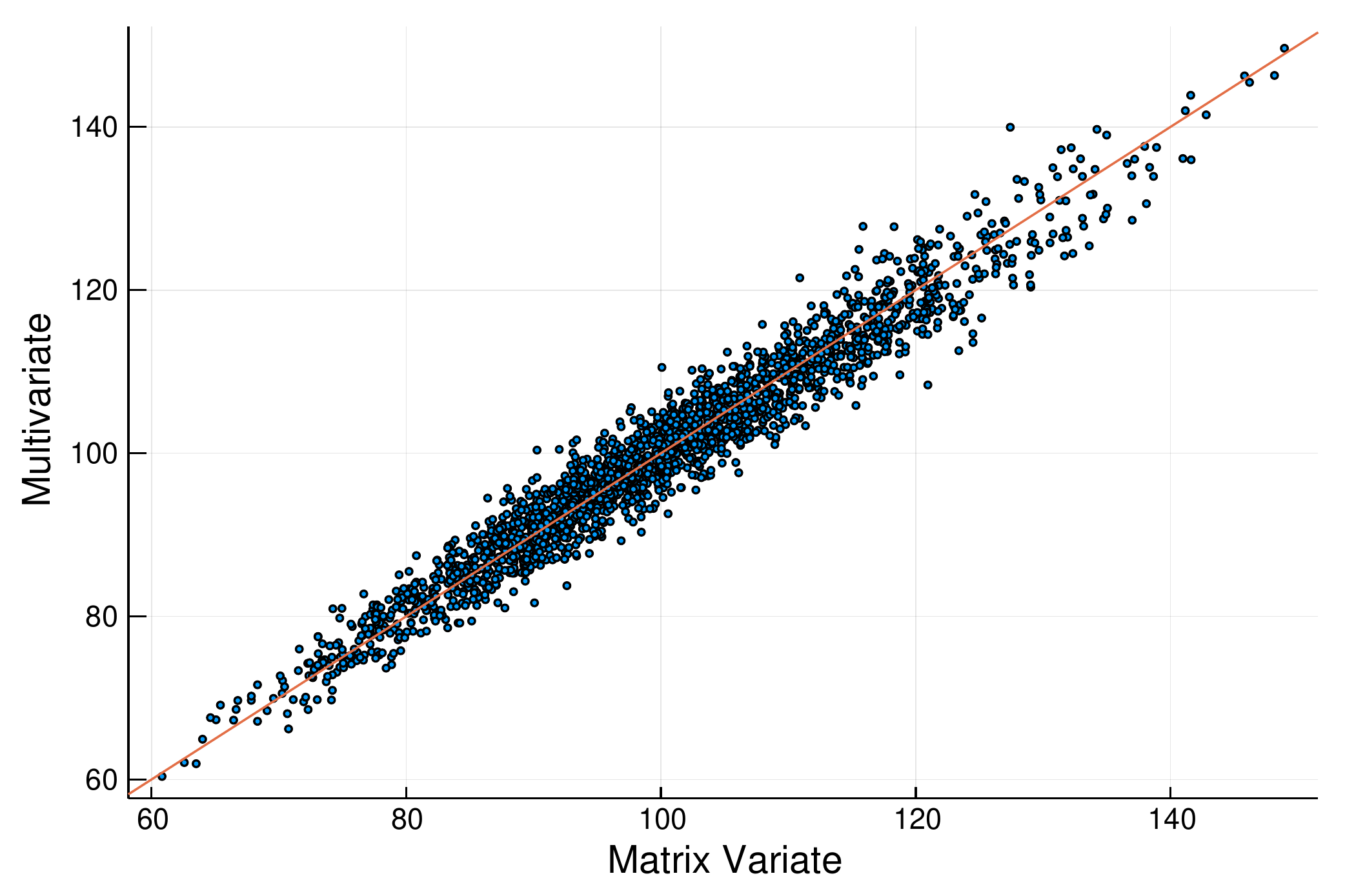}
        \label{fig:ddplot2000M}
    \end{subfigure}%
    ~ 
    \begin{subfigure}[t]{0.45\textwidth}
        \centering
        \caption{$N=2000, p=100$}
        \includegraphics[width=\textwidth]{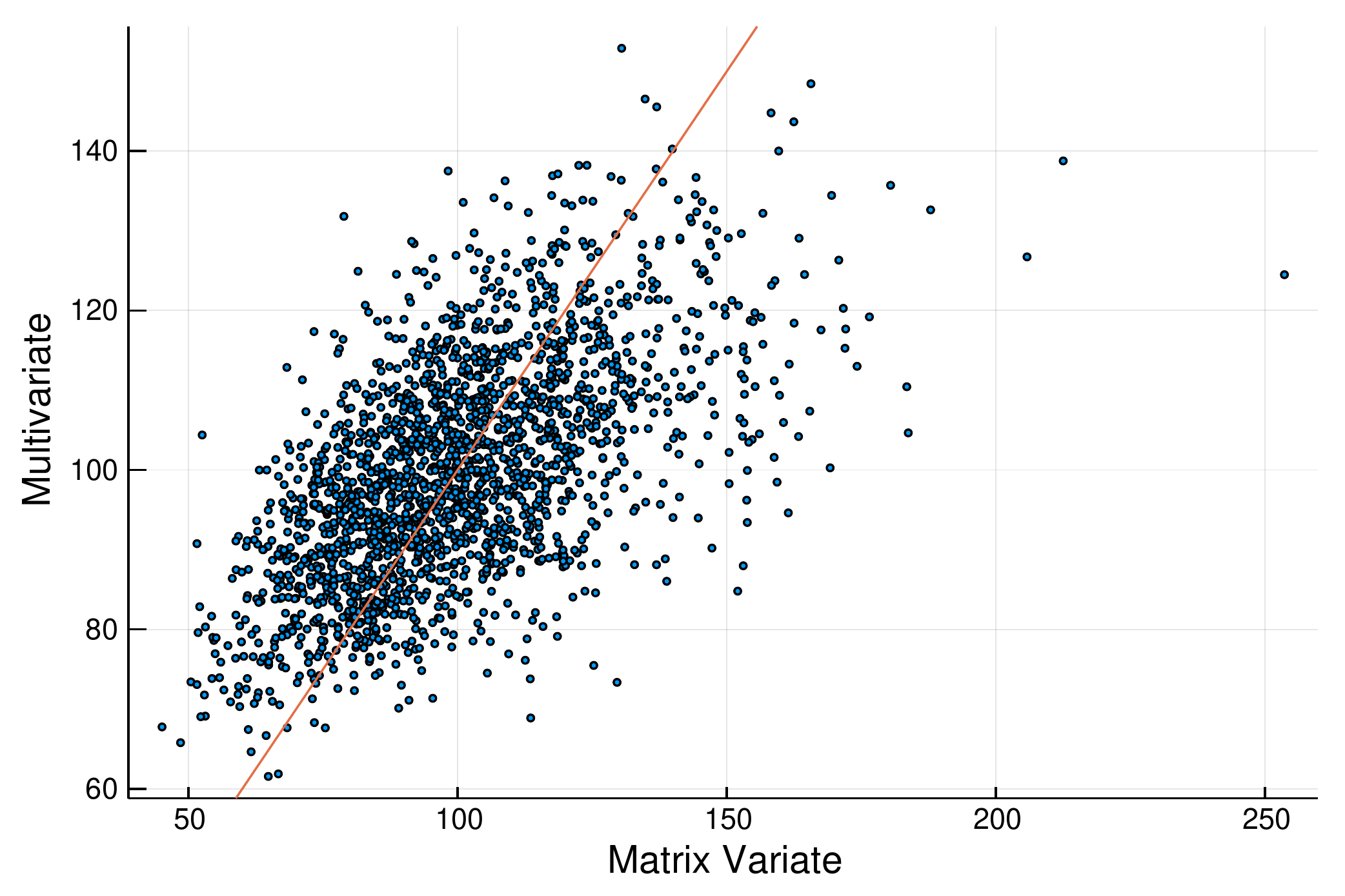}
        \label{fig:ddplot2000}
    \end{subfigure}
    \vspace{-0.15in}
    
    \begin{subfigure}[t]{0.45\textwidth}
        \centering
        \caption{$N=10000, p=100$}
        \includegraphics[width=\textwidth]{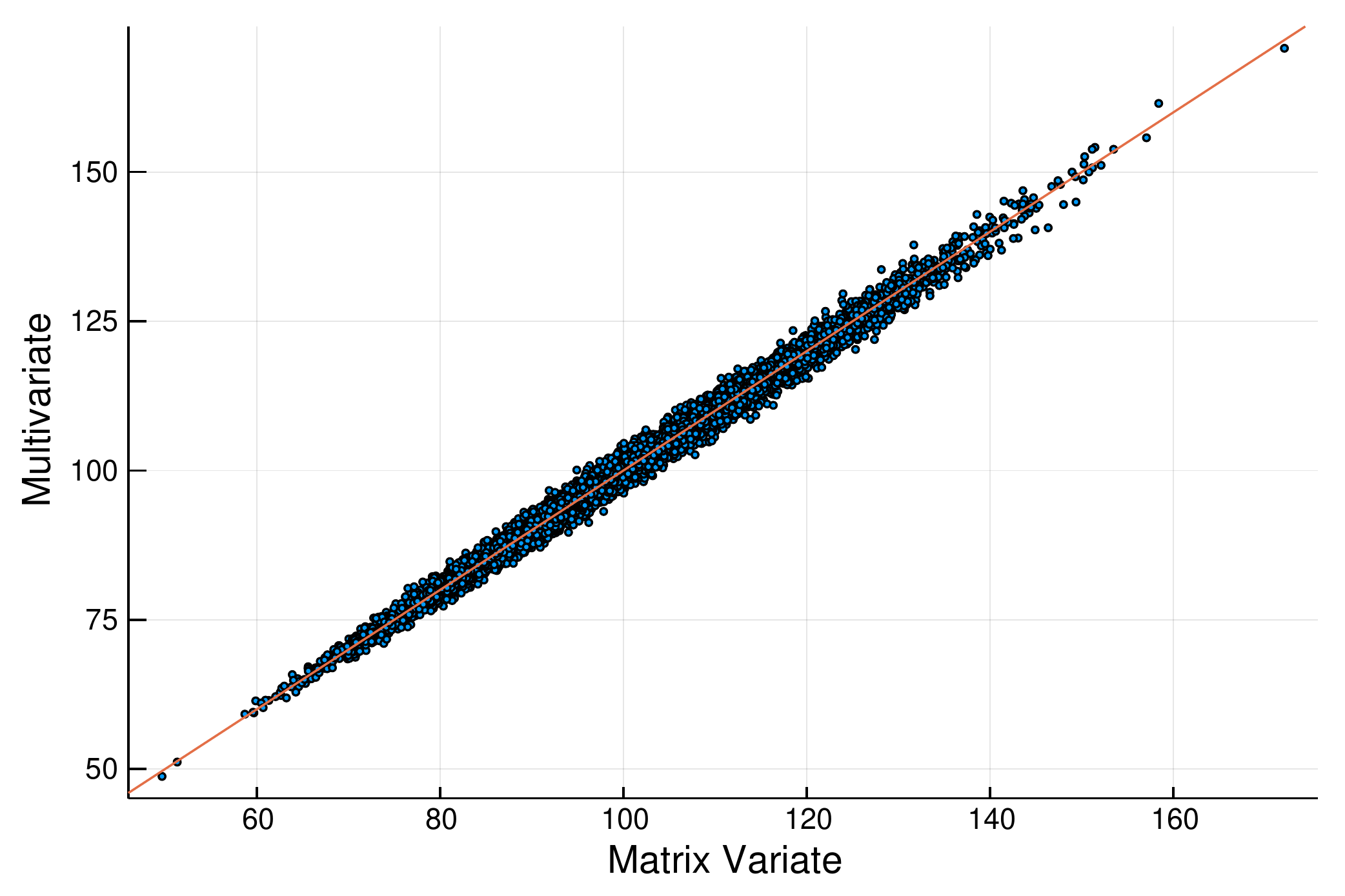}
        \label{fig:ddplot10000M}
    \end{subfigure}%
    ~ 
    \begin{subfigure}[t]{0.45\textwidth}
        \centering
        \caption{$N=10000, p=100$}
        \includegraphics[width=\textwidth]{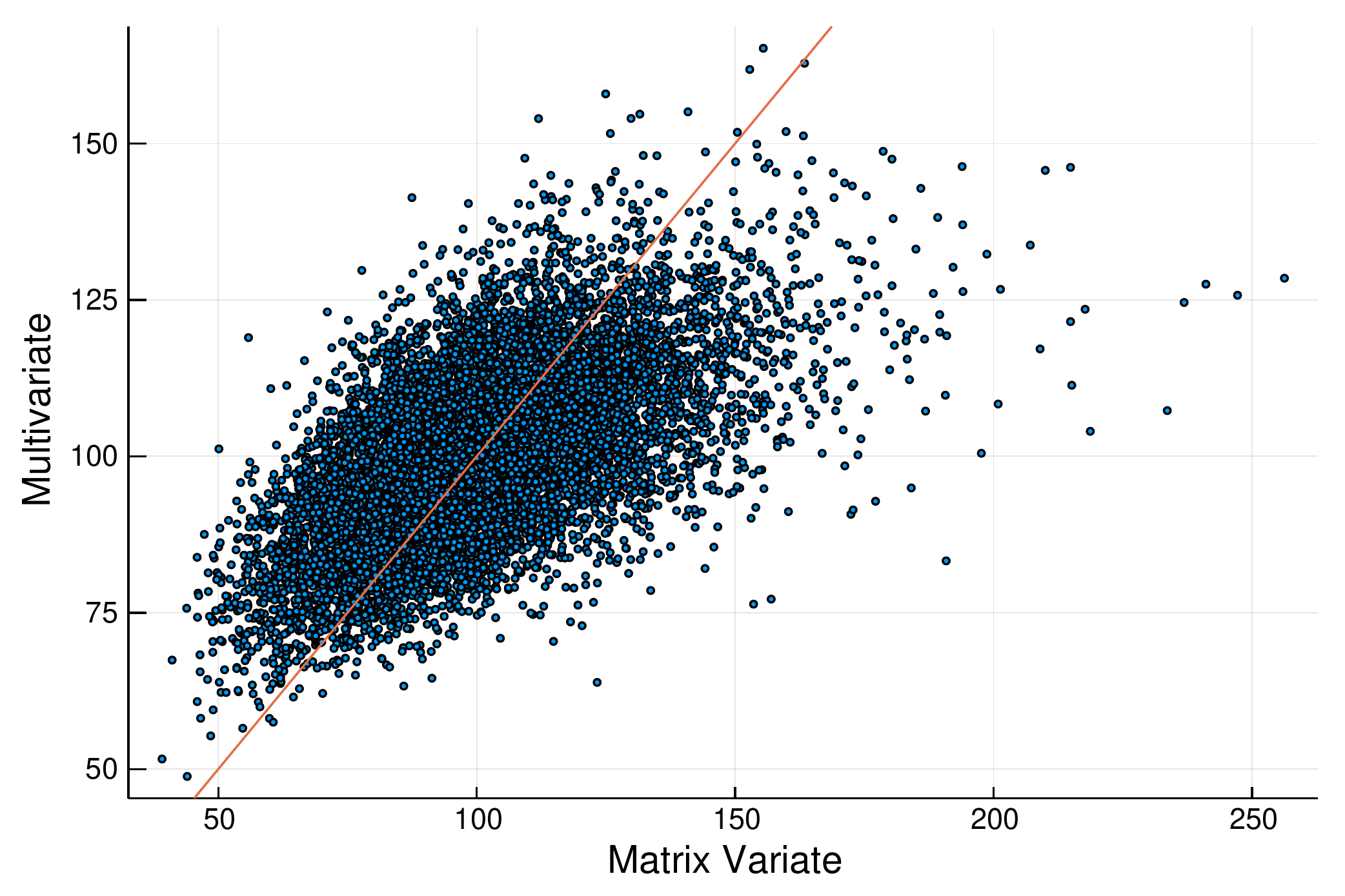}
        \label{fig:ddplot10000}
    \end{subfigure}
    \vspace{-0.15in}
        
    \caption{DD plots for simulated data with $N\in \{500,2000,10000\}$ and $p=100$, where matrix variate normal structure is present (left) or absent (right).}
    \label{fig:Nvary}
\end{figure*}

\subsection{Simulation study for KS test}

An investigation into the viability of the KS test for matrix variate normality is performed under various settings. The simulation study is conceived as follows. For a particular sample size $N$, and dimension $p$, $500$ randomly generated matrix variate datasets are subjected to the KS test. We are concerned with the Type 1 error i.e. the rejection of matrix variate normality when it is indeed present in the dataset. Out of the $500$ datasets of each sample size $N$ and dimension $p$, we simply calculate the proportion of times we reject the null hypothesis. We then further increase the sample size by $5$ observations and repeat the same experiment over all specified settings of dimensions. Figure~\ref{TYPE1} shows the performance of the KS test for various settings. We see as the sample size increases, Type 1 error is minimized. Furthermore we observe as the dimension increases there is a need for a larger sample size for the test to maintain performance. 
\begin{figure}[!ht]
\begin{center}
\includegraphics[scale=0.6]{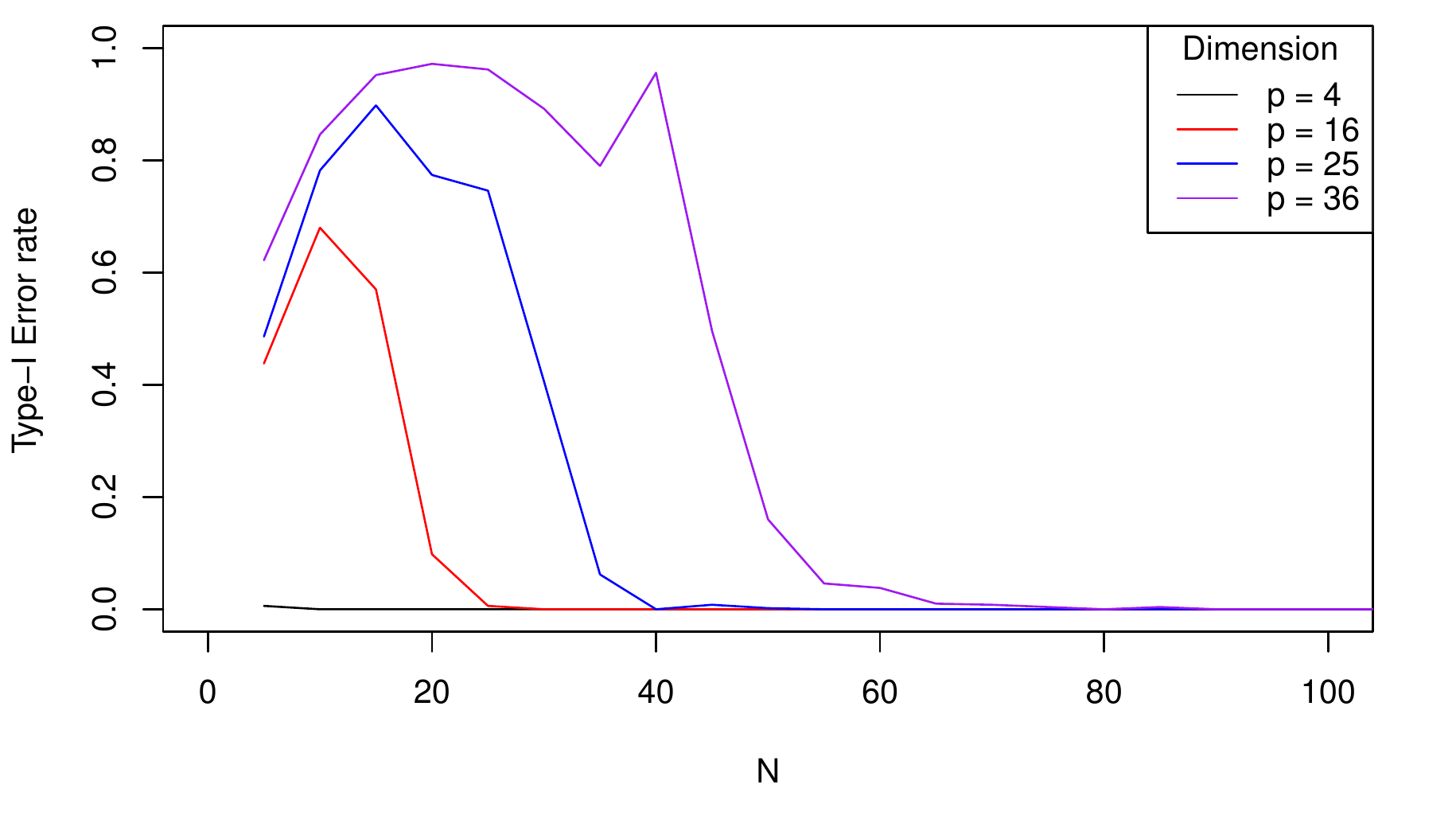}
\vspace{-0.2in}
\caption{Type 1 error measured for varying dimension and sample sizes.}
\label{TYPE1}
\end{center}
\end{figure}

\section{MNIST Handwritten Digits}

Our goal here is to consider whether the assumption of matrix variate normality holds for the MNIST handwritten digits dataset \citep{lecun98}. Note that the sample size for each digit class is sufficiently large, with $N>5000$ observations, and thus should be quite sufficient in estimating all parameters. Suppose each image is matrix variate normal distributed with appropriately sized dimension parameters $28\times28$. Parameter estimation can be performed using the estimates \eqref{e1}, \eqref{e2}, \eqref{e3}. In Figure~\ref{fig:mnistDD}, we consider the DD plots of digits $2$, $8$, and mixtures of $3$, $7$, and $1$. We report that the DD plots exhibit a pattern that does not resemble any of the plots from the simulation. In fact most of the plots exhibit extreme skewed distances such as in Figure \ref{fig:mnist2}.  When attempting to estimate parameters for a single distribution under a mixture setting, we notice the same skewed structure for all three digit classes in Figure \ref{fig:mnist37} and  \ref{fig:mnist71} . We also report that the KS test rejects the null hypothesis for all of the aforementioned settings in Figure~\ref{fig:mnistDD}. Given all the evidence, we conclude that the assumption of matrix variate normality is violated for the MNIST dataset. Interestingly, previous work has shown that skewed matrix variate models outperform the matrix variate normal approach in clustering and classification of these data \citep{mpb}.
\begin{figure*}[!htb]
\centering
	  \begin{subfigure}[t]{0.45\textwidth}
        \centering
        \caption{Digit $2$}
        \includegraphics[width=\textwidth]{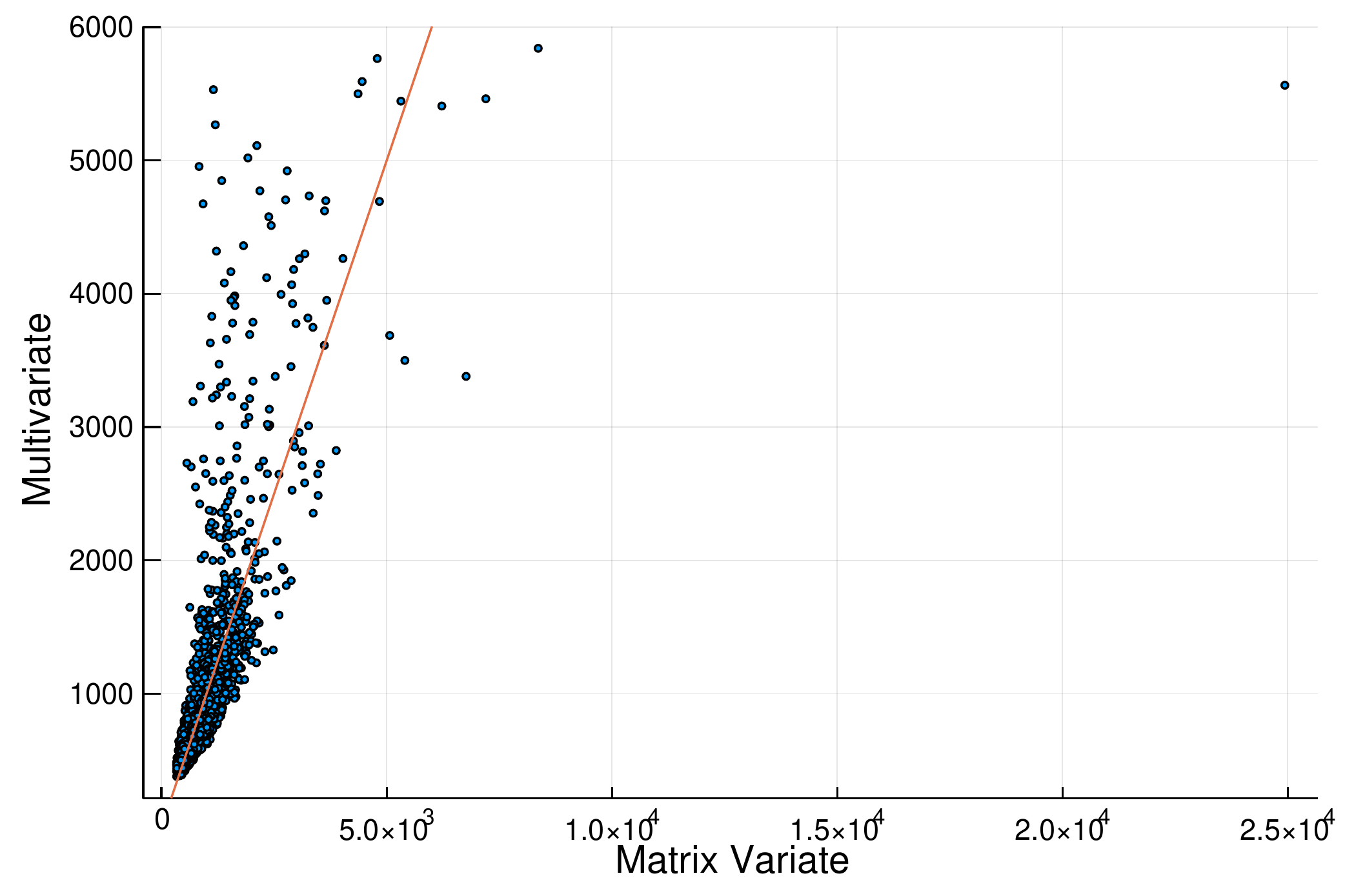}
        \label{fig:mnist2}
    \end{subfigure}%
        ~ 
    	  \begin{subfigure}[t]{0.45\textwidth}
        \centering
        \caption{Digit $8$}
        \includegraphics[width=\textwidth]{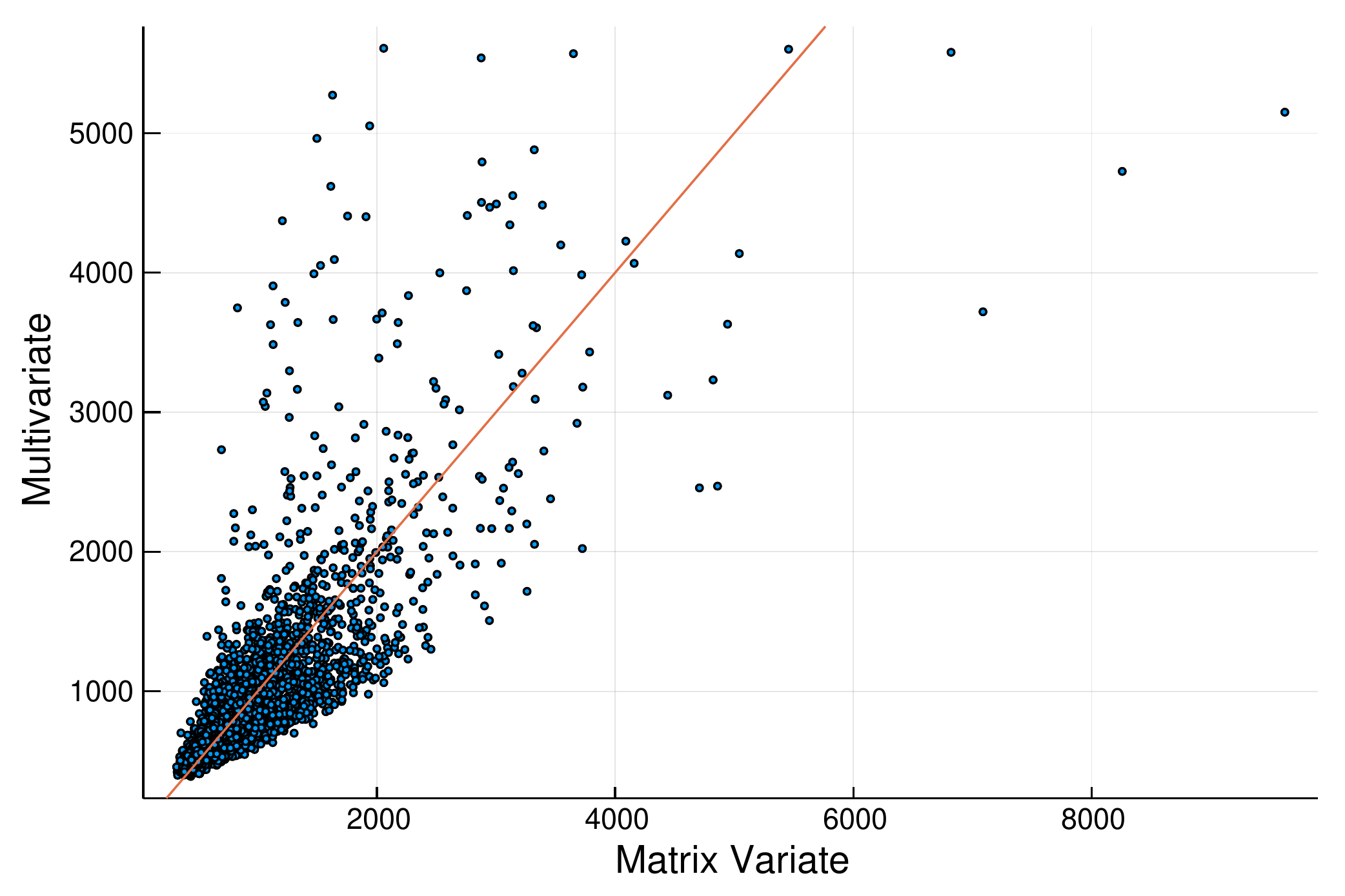}
        \label{fig:mnist8.png}
    \end{subfigure}\newline
    \vspace{-0.15in}
        
    \begin{subfigure}[t]{0.45\textwidth}
        \centering
        \caption{Mix of digits $3$ (blue) and $7$ (red)}
        \includegraphics[width=\textwidth]{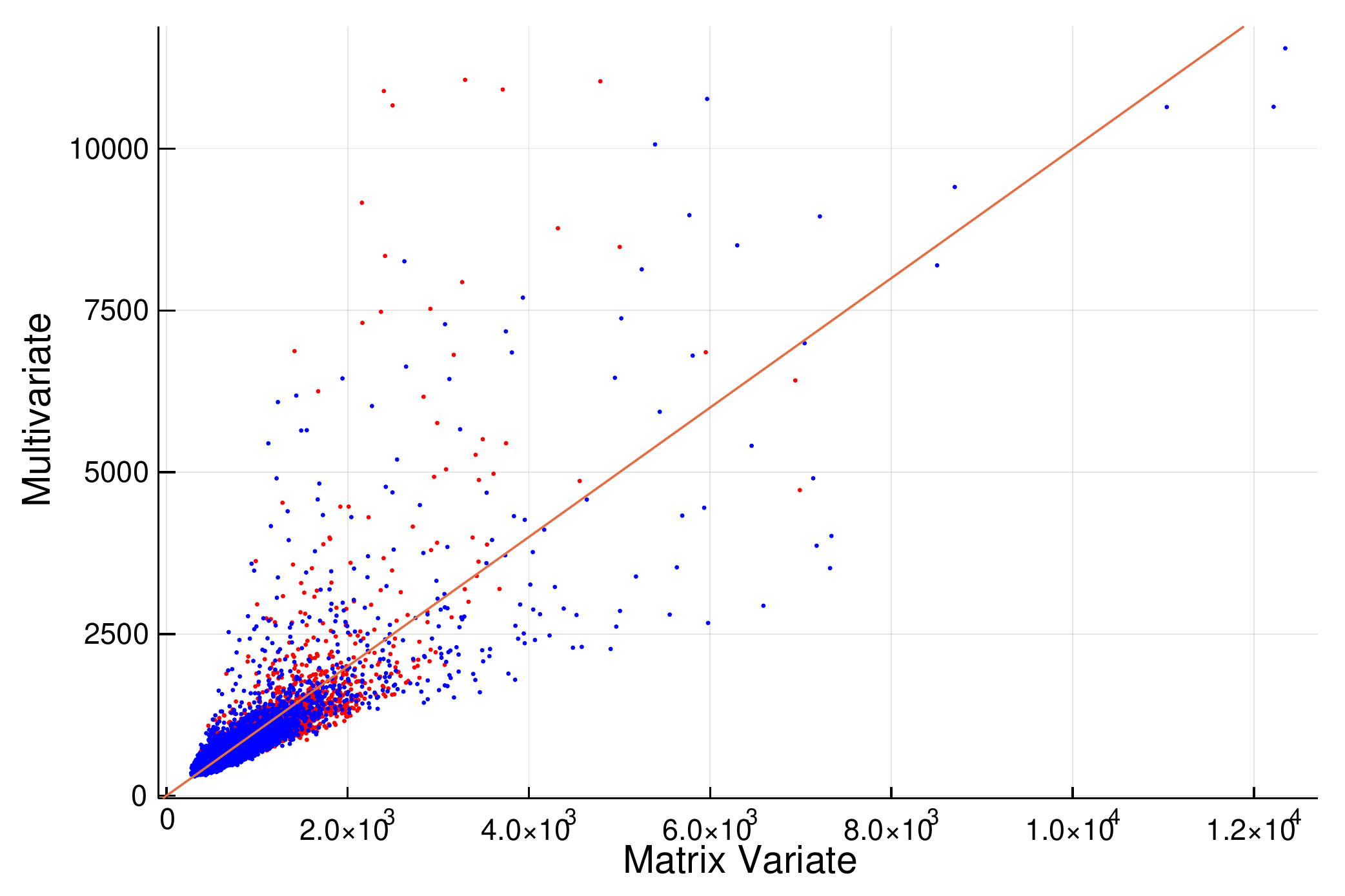}
        \label{fig:mnist37}
    \end{subfigure}%
    ~ 
    \begin{subfigure}[t]{0.45\textwidth}
        \centering
        \caption{Mix of digits $3$ (red)  , $7$ (blue), $1$ (green)}
        \includegraphics[width=\textwidth]{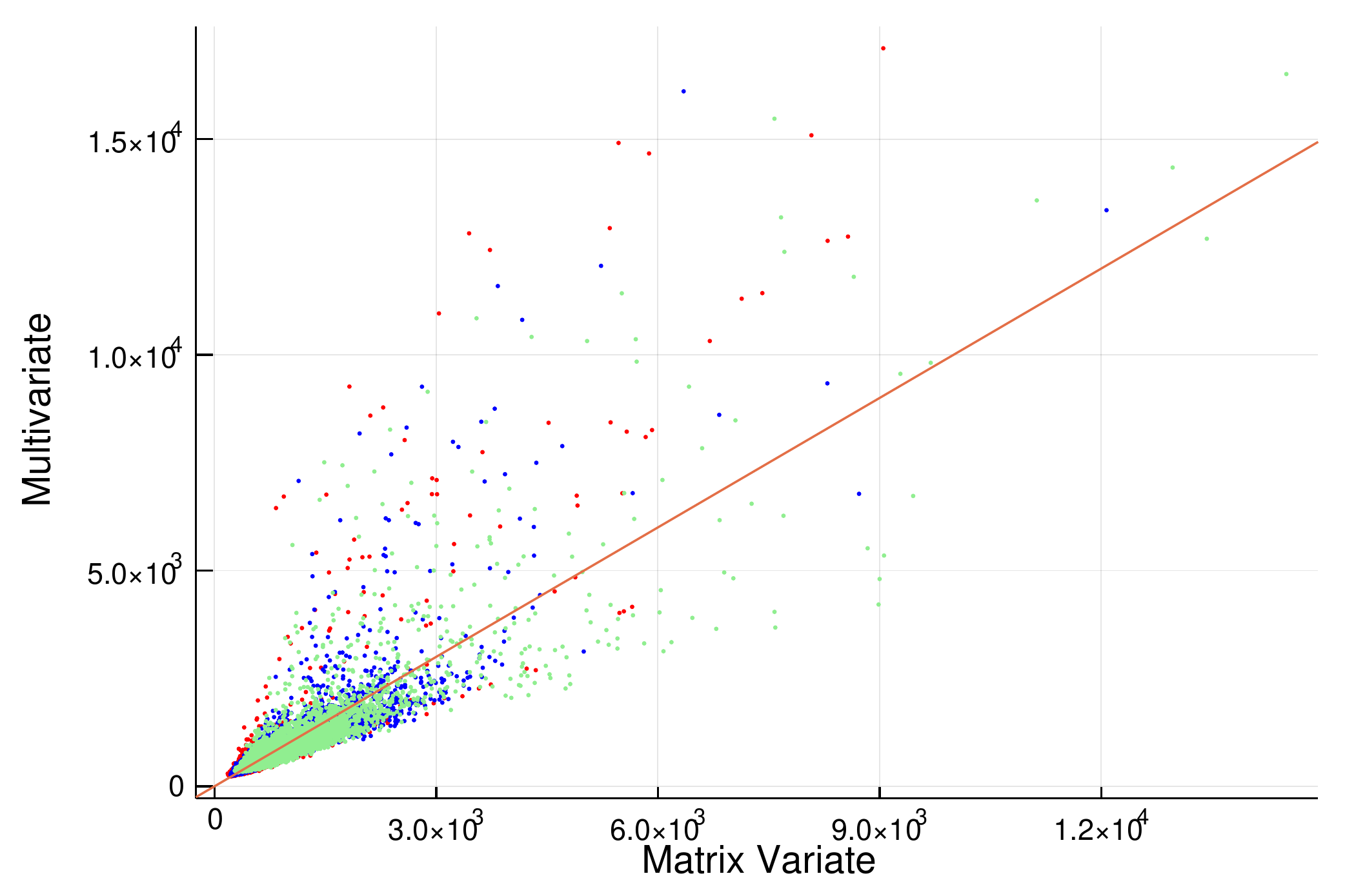}
        \label{fig:mnist71}
    \end{subfigure}
    \vspace{-0.15in}
        
        \caption{DD plots for MNIST dataset for several settings and digits.}
    \label{fig:mnistDD}
 \end{figure*}	

\section{Summary}
A framework for assessing matrix variate normality, both visually and using a statistical test, has been introduced. The new graphical technique for assessing matrix variate normality, called the DD plot, is based on comparing the respective MSDs. The DD plot was shown to be effective for assessing matrix variate normality for various dimensions and sample sizes. A KS test was also considered for a test of matrix variate normality, and we conducted a simulation to assess the Type 1 error for different sample sizes and dimensions. When applied to the MNIST example, the DD plots suggested that the assumption of matrix variate normality is violated. Moreover, the KS test rejected the hypothesis of matrix variate normality for each of the settings in Figure~\ref{fig:mnistDD}. This is not too surprising as the results in \cite{mpb} suggest the use of matrix variate skewed distributions results in better performance for this dataset. In conclusion, the DD plots along with the KS test constitutes a powerful combination for assessing matrix variate normality. The main avenue for future work is to extend the MSD to the space of skewed matrix variate data.
{\small

\section*{Acknowledgements}
This work was supported by a Vanier Canada Graduate Scholarship (Gallaugher), a Canada Graduate Scholarship (Clark), the Canada Research Chairs program (McNicholas), and an E.W.R~Steacie Memorial Fellowship.


\appendix
{\small
\section{Continuous Mapping Theorem}\label{app:cmt}
The continuous mapping theorem was first proved in 1943 and is sometimes referred to as the Mann-Wald theorem \citep{contMap}.
\begin{theorem}\label{thm:cmt}
Let $\{X_N\}$, $\{Y_N\}$, $X$, and $Y$ be random elements on some metric space $S$. In addition, let $g$ be a bivariate continuous map from one metric space $S$ to another $S'$. Then, 
$$X_N,Y_N \overset{P}{\longrightarrow} X,Y \ \ \Rightarrow \ \ g(X_N,Y_N) \overset{P}{\longrightarrow} g(X,Y).$$ 
\end{theorem}

\section{Relationship Between $\mathcal{D}$ and $\mathcal{D}_M$}\label{app:trivial}
Let $\bm{\mu} =  \text{vec}({\MM}) $ and  $\bm{\Sigma} =  \MV \otimes \MU$, then
\begin{align*}
\mathcal{D}(\bm{X}_i, \bm{\mu}, \bm{\Sigma} )   
&=   \big(\text{vec}(\bm{X}_i) -  \bm{\mu}\big)^{\top} \bm{\Sigma}^{-1} \big( \text{vec}(\bm{X}_i) -  \bm{\mu}\big)\\
&= \left\{\text{vec}(\bm{X}_i) -  \text{vec}({\MM})\right\}^{\top} (\MV \otimes \MU)^{-1} \left\{\text{vec}(\bm{X}_i) -  \text{vec}({\MM})\right\}\\
&= \text{vec}\left(\bm{X}_i - \MM\right)^{\top}
\left(\MV ^{-1}\otimes\MU^{-1}\right)\text{vec}\left(\bm{X}_i - \MM\right) \\
&= \text{vec}\left(\bm{X}_i - \MM\right)^{\top}
\text{vec}\left\{\MU^{-1} (\bm{X}_i - \MM)\MV ^{-1}\right\}\\
&= \text{tr}\left\{ \MV ^{-1} (\bm{X}_i - \MM)^{\top} \MU^{-1} (\bm{X}_i - \MM) \right\} \\
& = \text{tr}\left\{\MU^{-1}(\bm{X}_i - \MM)\MV^{-1}(\bm{X}_i - \MM)^{\top} \right\} \\
& 
= \mathcal{D}_M(\bm{X}_i, \MM, \MU, \MV) .
\end{align*}

}

\end{document}